\documentclass[preprint]{elsarticle}
\usepackage{rotating}
\usepackage{algorithm}
\usepackage{algorithmic}
\usepackage{amssymb}
\usepackage[usenames,dvipsnames]{xcolor}
\usepackage[colorlinks,citecolor=blue,linkcolor=BrickRed]{hyperref}
\usepackage{enumerate}
\usepackage{times}
\usepackage{amsfonts,latexsym,graphicx,epsfig,amssymb,color}
\usepackage{amsmath,amsthm,amstext,url,setspace}
\usepackage{amstext}    % defines the \text command, needed here
\usepackage{array}      % for advanced column specification: use >{\command} for
\usepackage{multirow}
\floatname{algorithm}{Random Process}
\newtheorem{theorem}{Theorem}[section]
\newtheorem{claim}[theorem]{Claim}
\newtheorem{proposition}[theorem]{Proposition}
\newtheorem{lemma}[theorem]{Lemma}

\newtheorem{observation}[theorem]{Observation}

\theoremstyle{definition}

\newtheorem{definition}[theorem]{Definition}

\renewcommand{\P}{\mathbf{P}} 
\newcommand{\NP}{\mathbf{NP}}

\newcommand{\barr}[1]{\overline{#1}}

\newcommand{\ignore}[1]{}
\newcommand{\cut}[1]{}
\newcommand{\sa}{\textsf{Sherali-Adams}}
\newcommand{\ls}{\textsf{Lov\'asz-Schrijver}}
\newcommand{\la}{\textsf{Lasserre}}
\newcommand{\lap}{\textrm{L\&P}} 
\newcommand{\iLS}{\textsf{LS}}  
\newcommand{\iLSp}{\textsf{LS}$_+$}
\newcommand{\iSA}{\textsf{SA}}
\newcommand{\iSAp}{\textsf{SA}$_+$}
\newcommand{\iLa}{\textsf{La}}
\def\pvc{$t$\textsc{-PVC}}
\def\bone{\textbf{1}}
\def\vc{\textsc{VC}}
\newcommand{\graph}{G=(V,E)}
\newcommand{\powerset}{\mathcal{P}}
\newcommand{\pr}[2]{\mathop{\mathbb P}_{#1}\displaylimits \left[#2\right]}
\newcommand{\Exp}[2]{\mathop{\mathbb E}_{#1}\displaylimits \left[#2\right]}
\def\distr{\mathcal{D}}

\def\psd {\succeq}

\def\bzero {{\bf 0}}
\def\be {{\bf e}}
\def\bv {{\bf v}}
\def\bb {{\bf b}}
\def\bd {{\bf d}}
\def\bq {{\bf q}}
\def \reals {\mathbb{R}}
\def \naturals {\mathbb{N}}

\newcommand{\Z}{\barr{\mathcal{Z}}}
\tnotetext[t1]{A preliminary part of this work appeared in the 34th Foundations of Software Technology and Theoretical Computer Science (FSTTCS 14), \cite{GeorgiouL14}.}
\fntext[fn1]{Research supported by NSERC.}
\fntext[fn3]{Research supported by an NSERC Undergraduate Student Research Award.}
\fntext[fn2]{Research supported by the Fields Institute Undergraduate Student Research Program.}
\cortext[corref1]{Corresponding authour.}

\begin{document}

\title{\bf Lift \& Project Systems Performing on the Partial-Vertex-Cover Polytope\tnoteref{t1}}

%\author{Jurek Czyzowicz \and Leszek Gasieniec \and Konstantinos Georgiou \and Evangelos Kranakis \and Fraser MacQuarrie}
%\author{Konstantinos Georgiou}
%\author{Edward Lee}

%\author{
%Konstantinos Georgiou\footnotemark[4]
%\and
%Edward Lee\footnotemark[4]~~\footnotemark[2] \\
%\vspace{-0.4cm}
%\and Department of Combinatorics \& Optimization, University of Waterloo \\
%\and Waterloo, Ontario, N2L 3G1, Canada 
%}

\author[1]{Konstantinos Georgiou\fnref{fn1}}
\ead{konstantinos@ryerson.ca}
\address[1]{Department of Mathematics, Ryerson University, 350 Victoria St.  Toronto, ON, M5B 2K3  Canada}

\author[2]{Andy (Jia) Jiang\fnref{fn2}}
\ead{jia.jiang@seh.ox.ac.uk}
\address[2]{University of Oxford, Wellington Square, Oxford, OX1 2JD  United Kingdom}

\author[3]{Edward Lee\corref{corref1}\fnref{fn3}}
\ead{e45lee@uwaterloo.ca}
\address[3]{Department of Combinatorics and Optimization, University of Waterloo, 200 University Ave.  West,  Waterloo, ON  N2L 3G1  Canada}

\author[4]{Astrid A. Olave\fnref{fn2}}
\ead{aaolaveh@unal.edu.co}
\address[4]{Universidad Nacional de Colombia,  Carrera 45 \# 26-85  Edif. Uriel Guti\'{e}rrez  Bogot\'{a} D.C.,  Colombia}

\author[5]{Ian Seong\fnref{fn2}}
\ead{seongi@carleton.edu}
\address[5]{Carleton College, One North College Street Northfield, MN 55057  U.S.A.}

\author[6]{Twesh Upadhyaya\fnref{fn2}}
\ead{twesh.upadhyaya@mail.utoronto.ca}
\address[6]{University of Toronto, 27 King's College Circle, Toronto, ON  M5S 1A1 Canada}

\begin{abstract}
We study integrality gap (IG) lower bounds on strong LP and SDP relaxations derived by the \sa\ (\iSA), \ls-SDP (\iLSp), \sa-SDP (\iSAp), and \la-SDP (\iLa) lift-and-project (\lap) systems for the $t$-Partial-Vertex-Cover (\pvc) problem, a variation of the classic Vertex-Cover problem in which only $t$ edges need to be covered. \pvc\ admits a 2-approximation using various algorithmic techniques, all relying on a natural LP relaxation. With starting point this LP relaxation, our main results assert that for every $\epsilon>0$, 
level-$\Theta(n)$ LPs or SDPs derived by all known \lap\ systems that have been used for positive algorithmic results (but the Lasserre hierarchy) have IGs at least $(1-\epsilon)n/t$, where $n$ is the number of vertices of the input graph. Our lower bounds are nearly tight, in that level-$n$ relaxations, even of the weakest systems, have integrality gap 1.   Additionally, we give a $O(\sqrt{n})$ integrality gap for the Level-1 Lasserre system and a superconstant general integrality gap for all Level-$\Theta(n)$ Lasserre derived SDPs.

As lift-and-project systems have given the best algorithms known for numerous combinatorial optimization problems, our results show that restricted yet powerful models of computation derived by many \lap\ systems fail to witness $c$-approximate solutions to \pvc\ for any constant $c$, and for $t=O(n)$. 
%This is one of the very few known examples of an intractable combinatorial optimization problem, for which LP-based algorithms induce a constant approximation ratio, still lift-and-project LP and SDP tightenings of the same LP have unbounded IGs.
As further motivation for our results, we show that the SDP that has given the best algorithm known for \pvc\ has integrality gap $n/t$ on instances that can be solved by the level-1 LP relaxation derived by the \iLS\ system. This constitutes another rare phenomenon where (even in \cut{some} specific instances) a static LP outperforms an SDP that has been used for the best approximation guarantee for the problem at hand. 
Finally, we believe our results are of independent interest as they are among the very few known integrality gap lower bounds for LP and SDP 0-1 relaxations in which not all variables possess the same semantics in the underlying combinatorial optimization problem.
%Most importantly, one of 
To achieve our results, we utilize a common methodology of constructing solutions to LP relaxations that almost trivially satisfy constraints derived by all SDP \lap\ systems known to be useful for algorithmic positive results (except the \iLa\ system). The latter sheds some light as to why \iLa\ tightenings seem strictly stronger than \iLSp\ or \iSAp\ tightenings. 
\end{abstract}


\maketitle
 
\section{Introduction}

%\begin{definition}[The Partial Vertex Cover Problem (\pvc)] 
Let $\graph$\ be a graph on $n$ vertices 
%where every vertex $i$ is equipped with a weight $w_i \in \reals_+$, 
and $t \in\naturals$, with $t\leq |E|$. A subset of vertices $S$ that are incident to at least $t$ many edges is called a \textit{$t$-partial vertex cover}. In the $t$-Partial-Vertex-Cover (\pvc) optimization problem, the goal is to find a $t$-partial vertex cover $S$ of minimum size. %that minimizes the cost $w(S):=\sum_{i \in S} w_i$. 
\pvc\ is a tractable optimization problem whenever $t=\Theta(1)$\cut{, since one can simply enumerate in polynomial time all $\binom{|E|}{p}$ many subinstances of $G$ induced by $p$ many edges, and then solve optimally (using brute-force) each subinstance insisting in covering all edges, and finally choosing the best solution found}. In the other extreme, $|E|$\textsc{-PVC} is exactly the classic $\NP$-hard problem known as minimum Vertex-Cover (\vc). As such, any hardness of approximation for \vc\ translates to the same hardness for $|E|$\textsc{-PVC}. In particular, $|E|$\textsc{-PVC} is 1.36 and $(2-o(1))$ hard to approximate assuming $\P \not = \NP$~\cite{DS05} and the Unique Games Conjecture~\cite{KR03} respectively. Moreover, 
there exists an approximation preserving reduction from \pvc\ to \vc\ as long as $n/t=n^{\Theta(1)}$ \cite{BB98}.
%any $\alpha$ approximation for \pvc\ can be converted into an $\alpha$ aproximation for \vc\ as long as $n/t=n^{\Theta(1)}$ \cite{BB98}. 
Unlike \vc, \pvc\ is also known to be hard in bipartite graphs~\cite{CS13}.
On the positive side, \cite{Hoc98,Mes09,TDY14} have proposed 2-approximation algorithms even for the weighted version of \pvc\ (see~\cite{KPS11} for a wider family of results concerning partial covering problems). 
The common starting point of all these results is the standard 0-1 LP relaxation for \pvc\ (see~\eqref{equa: LP pvc} in Section~\ref{sec: problem def and LP relaxation}).
The best (asymptotic) approximation known for \pvc\ relies on a SDP relaxation and achieves a $2-\Omega\left( \log \log n / \log n\right)$ ratio~\cite{HS02}. 

\cut{, \pvc\ is known to be fixed parameter tractable~\cite{Bla03}, while the best exact algorithm requires $O^*(1.396^t)$ steps~\cite{KLR08}. As for its approximability, which is also the focus of the current work,}
\cut{ FIX AFTER RECOVER ...... ignore the 2014 publication which is not reviewed yet. 
 starting with~\cite{Hoc98} that uses a Lagrangian-based attack, \cite{Mes09} that uses a primal-dual approach, and \cite{TDY14} that uses iterative rounding.}

\cut{
As a result, \pvc\ is yet another example of an intractable combinatorial optimization problem for which convex-programming based algorithms have given the best algorithms known. 
}

A standard performance measure for convex-programming (LP or SDP) relaxations is the so-called integrality gap (IG), i.e. the worst possible ratio 
between the cost of the exact optimal solution and the cost of the relaxation.
As a measure of complexity,  IG upper or lower bounds are informative for two main reasons: (1) the majority of convex-programming based approximation algorithms attain an approximation ratio equal to the best provable upper bound on the IG. 
(2) Convex-programming relaxations can be seen as a restricted and static model of computation that can immediately witness (using fractional solutions) the existence of good (integral and) approximate solutions, without even finding them.

In this direction, it is notable that for a long series of combinatorial optimization problems, the best approximability known agrees with the IG of natural convex-programming relaxations, \cite{Rag08} being the most notable example. 
\cut{
A characteristic example is \vc, for which Halperin's SDP~\cite{Hal02} has IG at most $2-\Omega( \log \log \Delta / \log \Delta)$ (where $\Delta$ is the maximum vertex degree of the input graph), which is the best possible approximability (with respect to the asymptotic improvement upon constant 2) modulo the Unique Games Conjecture~\cite{AKS11}.
}
In contrast, all analyses for convex-programming relaxations for \pvc\ \cut{\cite{Hoc98,Mes09,TDY14,HS02},}
\cite{HS02,Hoc98,Mes09}
witness some integral solution with cost $sol$ to the relaxation satisfying $sol \leq 2 \cdot rel + \Theta(1)$, where $rel$ is the value of the relaxation. Note that this leaves open the possibility that the IG of these relaxations is unbounded when the optimal solution has small enough cost. In fact, it was already known that 
\cut{even for the unweighted \pvc,}
the standard 0-1 relaxation \eqref{equa: LP pvc} has IG at least $n/t$ 
%(where $n$ is the number of vertices in the input graph), 
while we also establish the same IG for the SDP of~\cite{HS02}.
\cut{
This by itself constitutes an interesting and rare phenomenon, in which the best approximate solution-costs, that natural convex-programming relaxations can witness, are far off from the best approximate solutions achieved by combinatorial algorithms (that could even rely on the relaxations). 
}

Very interestingly, the power of convex-programming for combinatorial optimization problems is not limited by the performance of the natural and static relaxations. A number of systematic procedures, known as lift-and-project (\lap) systems, 
%e.g. that of \ls\ (\iLS\ and \iLSp)~\cite{LS91}, \sa\ (\iSA\ and \iSAp)~\cite{SA90} and \la\ (\iLa)~\cite{Las01}, 
have been proposed in order to reduce the IG of 0-1 LP relaxations $P \subseteq [0,1]^m$ (the reader should think of $P$ as the feasible region of a relaxation of some combinatorial problem).
%\cut{, e.g. the feasible region of relaxation~\eqref{equa: LP pvc} for \pvc).}
%for an instance $\graph$\ of \pvc, $P$ is the feasible region of the natural LP relaxation 
% $P_t(G)$ of \eqref{equa: LP pvc}, and $[m]=V\cup E$.
%which may have a large integrality gap, 
%deriving a sequence of nested relaxations with useful algorithmic properties. 
The seminal works of Lov\'asz and Schrijver~\cite{LS91}, Sherali and Adams~\cite{SA90}, Lasserre~\cite{Las01}, and Parrilo~\cite{parrilo2003semidefinite} give such systematic methods (\iLS, \iLSp, \iSA, and \iLa\ respectively).\footnote{\iLSp\ and \iSA\ systems derive stronger relaxations than the \iLS\ system, while \iLSp, \iSA\ are incomparable. \iLa\ derives SDPs that are at least as strong as relaxations derived by any other system.} Starting with the polytope $P$, each of the %\iLS, \iLSp\ and the \iSA\ 
systems derives a sequence (hierarchy) of relaxations $P^{(r)}$ for $P \cap \{0,1\}^m$ 
\cut{(corresponding to solutions to the combinatorial problem at hand)
satisfying
$
P \subseteq P^{(1)} \subseteq P^{(2)} \subseteq \ldots \subseteq P^{(m)} = conv\left( P \cap \{0,1\}^m \right).
$
That is, the sequence of relaxations $P^{(r)}$ }
that are nested, preserve the integral solutions of $P$, and $P^{(m)}$ is exactly the integral hull of $P$ (hence the IG of the last relaxation is 1 independently of the underlying objective). For these reasons, these systems are also known as hierarchies (of LP or SDP relaxations).  More importantly, if $P$ admits a (weak) separation oracle, then one can optimize a linear objective over the so-called level$-r$ relaxation $P^{(r)}$ of all methods but the \iLa\ system in time $m^{O(r)}$. For the \iLa\ system, one requires stronger conditions, see~\cite{o2017sos,MR3685820}.

In other words, all \lap\ %\iLS, \iLSp\ and the \iSA\ 
systems constitute ``parameterized'' models of computation for attacking intractable combinatorial optimization problems. 
Even more interestingly, there are numerous combinatorial problems for which either \lap\ systems have given the best approximation algorithms known (with no matching combinatorial algorithms known), or with approximation guarantees matching the best combinatorial algorithms known. We refer the reader to~\cite{CM12} for a relatively recent survey. 

For this reason, a long line of research has been devoted in proving IG lower bounds for relaxations derived by \lap\ systems, while any such result is understood as strong evidence of the true inapproximability of the combinatorial problem at hand. %see for example .... for some recent examples on the Lasserre system. 
At the same time, an $\alpha$ IG for level-$r$ relaxations derived by \lap\ systems implies that algorithms (for a restricted yet powerful model of computation) that run in time $m^{O(r)}$ cannot witness the existence of $\alpha$-approximate solutions to the combinatorial problem. It is notable that examples of integrality gaps for \lap\ systems that are way off from the best approximability known for a combinatorial optimization problem are quite rare.

\subsection{Our contributions \& Comparison to previous work}

\cut{
\subsection{Comparison to previous integrality gap lower bounds}
}

To the best of our knowledge, this is the first study of integrality gap lower bounds for lift-and-project tightenings of the natural 0-1 relaxation of \pvc. Our starting point is the standard LP relaxation~\eqref{equa: LP pvc} that has been used in all 2-approximation algorithms for weighted instances. Our goal is to derive strong integrality gap lower bounds for level-$r$ relaxations derived by the \iLSp, \iSA\ and \iSAp\ systems, where $r$ is as large as possible, and $t=O(n)$ (where $n$ is the number of vertices in the input graph). It is worthwhile noticing that there is a number of very strong IG lower bounds known for \vc\ in \lap\ systems, including IG of $2-\epsilon$, for every $\epsilon>0$, for 
level-$\Theta(n)$ \iLS\ LPs~\cite{STT07b},
level-$n^{\Theta(1)}$ \iSA\ LPs~\cite{CMM09},
level-$\Theta(\sqrt{\log / \log \log n})$ \iLSp\ SDPs~\cite{GMPT10},
level-$5$ \iSAp\ SDPs~\cite{BCGM11},
and IG of $7/6-\epsilon$ and 1.36 for level-$\Theta(n)$~\cite{Sch08} and level-$n^{\Theta(1)}$~\cite{Tul09} \iLa\ SDPs. Each of the aforementioned lower bounds %(some for stronger systems than the ones considered here), 
imply directly the same IG lower bounds, for the same level relaxation and for the same system for~\eqref{equa: LP pvc} by a straightforward reduction. However for the magnitude of $t$ for \pvc\ for which we establish our results (roughly speaking for $t\leq n/2$), and in which the problem makes the transition from tractable to intractable, our IG lower bounds are superconstant and not just 2.

\cut{
The focus of our current work is the performance of \iLS, \iLSp, and \iSA\ systems on the standard relaxation~\eqref{equa: LP pvc} of \pvc, knowing that the IG of the latter relaxation is at least $n/t$ even for unweighted star-graphs~\cite{Mes09}. 
}
The majority of our results are negative. Our motivating observations are that (a) a simple graph instance is responsible for a $n/t$ IG of the SDP of \cite{HS02} (Proposition~\ref{prop: star fools halperins sdp}), on which the best algorithm know for \pvc\ is based and (b) the level-1 LP derived by the \iLS\ system (which is strictly weaker than the \iLSp\ and \iSA\ systems) solves the same instances exactly (Proposition~\ref{prop: LS1 solves star}). This is a remarkable example of a simple LP that outperforms, even in a specific instance, an SDP that has been used for the best algorithm for a combinatorial problem (the authors are not aware of another similar example). It is natural then to ask whether relaxations derived by \lap\ systems can witness existence of 2-approximate solutions to \pvc. We answer this question in the negative by proving strong IG lower bounds for all \lap\ systems that have been used for positive algorithmic results. For all these systems we show that as long as $n\geq 2r+ 2t+2$, the level-$r$ relaxations have integrality gap at least $\binom{n-2r}2/t\cdot n$. As an immediate corollary, we see that the integrality gap of the starting LP (which is at least $n/t$) remains $(1-\epsilon)\frac{n}t$ for level-$\Theta(n)$ LP and SDP relaxations. %, whenever $t=O(n)$.
Note that our results could be also stated as rank lower bounds of a certain knapsack-type inequality (the one certifying a good IG). Many similar results have appeared in the literature, e.g.~\cite{Che07,CD01,kurpisz2016hardest,kurpisz2016sum,Lau03}, but they are all for polytopes that are of different structure than the partial vertex cover polytope. 

\cut{
Our results indicate that \lap-relaxations fail to reduce the IG of \eqref{equa: LP pvc}, even though the level-$n$ relaxations have integrality gap 1. 

As a byproduct we exhibit a very rare example of a combinatorial optimization problem that admits a constant approximation, still a sequence of very strong convex-relaxations (LPs and SDPs) fail to witness the same approximability. Even stronger than this, 
}

The above negative results bring up another rare phenomenon; for the family of tractable combinatorial optimization problems \pvc, for which $t=\Theta(1)$, \lap-relaxations have unbounded discrepancy; see also~\cite{kurpisz2016hardest,kurpisz2017unbounded} for similar results. 
This is in contrast to many combinatorial optimization problems, and in particular \vc, for which constant-level \lap-relaxations either have integrality gaps matching the best approximability or they even solve tractable variations of the problems. Finally, due to the approximation preserving reduction from \vc\ to \pvc~\cite{BB98}, when $t=n^{\Theta(1)}$, our results also imply that \lap\ systems applied on the \pvc\ standard polytope cannot yield new insights for the $\NP$-hardness inapproximability of \vc.
However, the aforementioned integrality gap lower bounds are tailored to the specific formulation. In other words, it is still possible that \lap\ systems might induce much better integrality gaps, given a different LP relaxation for the underlying combinatorial optimization problems. 

Our contributions are twofold. First, we establish strong IG \lap\ lower bounds for LPs defined over two types of variables. Lower bounds for \lap\ relaxations of such polytopes are very rare (the authors are aware only of one such result~\cite{KM14}). 
Second, we utilize a generic condition of solutions to LP relaxations that can fool a large family of PSD constraints (for a high level explanation of the condition see Section~\ref{sec: techniques}).

\subsection{Our techniques}\label{sec: techniques}

%Achieving \lap\ IG lower bounds has been proven particular challenging, for polytopes that involve two types of linear variables such as in \eqref{equa: LP pvc}. 
For our main results we employ some standard and generic techniques for constructing vector solutions for convex relaxations derived by the \iSA\ system. Then we utilize a condition (which was uncommon prior to an early version of the current work, see~\cite{GeorgiouL14}) 
special to our solution that allows us to argue that the same construction is robust against SDP tightenings. Our IG instance is the unweighted clique on $n$ vertices, which for all $t$, admits an optimal solution of cost 1. This IG construction suffers a decay that is proportional to $\binom{n-2r}{2}$. The decay with $r$ is unavoidable, at a high level, due to that level-$r$ relaxations solve accurately local subinstances induced by $r$ many elements corresponding to variables. 
Since our LP relaxation has edge variables, the removal of $r$ many edges induces a clique of $n-2r$ vertices. Since we still have $\binom{n-2r}{2}$ edges, each edge needs to be covered ``on average'' $t/\binom{n-2r}{2}$ fractional times. Due to the symmetry imposed in our solutions, this is also the contribution of each vertex in the objective. %level-$r$ relaxation. 

\cut{
\subsubsection{Establishing the \iLS\ and \iLSp\ lower bounds} 
}

\textbf{Establishing the \iSA\ IG lower bound:}
\cut{
Achieving \iSA\ lower bounds has been proven particular challenging, especially for polytopes that involve two types of linear variables, as in \eqref{equa: LP pvc}. 
%(the only other such lower bound known to the authors is \cite{KM14}). 
}
A common and generic approach for constructing \iSA\ solutions is to use the probabilistic interpretation of the system, first introduced in~\cite{KV07}, and that is implicit in all our arguments of Section~\ref{sec: SA lower bound}. At a high level, the curse and the blessing of the \iSA\ system is that level-$r$ solutions are convex combinations of (LP feasible) vectors that are integral in any set of $r$ many variable-indices. These convex combinations can be interpreted as families of distributions of \textit{feasible} integral solutions for subsets of the input instance of size-$r$ (hence subsets of variables as well), that additionally enjoy the so-called \textit{local-consistency} property: distributions over different subinstances should agree on the solutions of the common sub-subinstance. Designing such probability distributions over sets of indices that also enclose the support of any constraint gives automatically a solution to the level-$r$ \iSA. Finding however such distributions is in general highly non trivial, especially when aiming for a big integrality gap. 

The previous recipe is not directly applicable to the \pvc\ polytope, as it has a defining facet that involves all edges of the input graph. This means that had we blindly tried to find families of probability distributions as described above, then we would have unavoidably defined distributions of feasible solutions in the integral hull. Our strategy is to deviate from the generic probabilistic approach, and focus first on satisfying constraints of the \iSA\ relaxation of relatively small support.
\cut{
(hence we will only need distributions over small sets of variable-indices). 
}

At a high level, the novelty of our approach is that we do not explicitly define locally consistent distributions of local 0-1 assignments, one for each subset of variables of bounded size, rather we achieve this implicitly. One of the advantages of our construction is that it is surprisingly simple. Specifically, we define a \textit{global} distribution of 0-1 assignments as follows: each of the vertices is chosen in the solution independently at random, and with negligible probability, and covered edges are those incident to at least one chosen vertex. 

The locally consistent distributions, that we need to associate each subset of variables $A$ with, are obtained by restricting the global distribution onto the subinstance induced by $A$. This trick can be thought as a vast generalization of the so-called correction-phase (or expansion recovery) that is common to all \iSA\ lower bounds, although it is sometimes hidden in the technicalities of the proofs (\cite{GM08} is a good example where the correction phase is made explicit). According to this trick, set $A$ is  effectively blown up (or ``corrected'') to a big enough superset $\barr A$ with certain structural properties. This allows for sampling almost uniformly at random over local 0-1 assignments (of variables in $\barr A$) that can be easily seen to induce consistent local distributions, whereas the same task seems to be impossible to be realized directly on $A$. Interestingly, $\barr A$ is the whole instance in our case. 

Our global distribution has a special property that it always satisfies all linear constraints of the \pvc\ polytope but the one demand-constraint, i.e. the constraint that requires $t$ many edges to be covered. In particular, the proposed vector solution is a convex combination of exponentially many solutions in the integral hull and of the outlier all-0 vector. In fact our global distribution assigns probability $1-o(1)$ to the latter vector, which is also responsible for the large integrality gap. 

Notably, there is no generic reason to believe that such a vector solution satisfies the almost global constraint of the \pvc\ polytope that involves all edges. To that end, we take advantage of the fact that we do \textit{not} need to define feasible solutions of the whole instance in every small subinstance. This means that if presented with a small subinstance of the input graph, we are allowed in principle to cover zero edges in that subgraph with positive probability, as long as we do cover $t$ many edges in the complement. That said, constraints of large support cannot be treated probabilistically with respect to the global distribution. Instead, we deal with such constraints almost algebraically (in contrast to the majority of \iSA\ consructions), as one would normally do for a standard LP. More specifically, we rely on the fact that when we condition on covering zero edges in a subclique of size at most $2r$, edges that do not touch this subclique are covered independently at random with significant probability compared to how many edges are left. Linearity of expectation then can prove for us that the demand constraint is indeed satisfied. 

\cut{
(We emphasize that the norm is to show \iSA\ lower bounds probabilistically, with a notable exception~\cite{CFG13}, as well as ~\cite{MS09,CGGS13} if the proofs therein are seen under a broader lens.) In that sense, our approach can be seen as a hybrid between the common probabilistic approach and the uncommon algebraic approach. 

Finally, all these ingredients need to be fine tuned to also give a large integrality gap. One of the advantages of our \iSA\ lower bound is that it is elegantly simple. The family of distributions that achieve all the above is the following; For any set $A$ of at most $r$ many vertices and edges, consider some subclique $Q_{2r}$ containing $A$. We decide to cover each edge of $Q_{2r}$ independently with probability $p$  (that needs to be carefully chosen), and for each edge that is to be covered, we do pay for both its endpoints. 
%Choosing then $p$ carefully can prove the promised integrality gap, as long as each vertex is chosen in the solution with negligible probability. 
As a side note, such a solution is much different than the vector solutions for \iLS, and \iLSp, in that it assigns to edge variables and vertex variables different values. 
}

\textbf{Establishing IG lower bounds for SDP hierarchies:} Showing that our \iSA\ vector solution is robust against SDP tightenings is by construction very easy. The reason is that all SDP hierarchies (that have been used for positive algorithmic results), except the \iLa\ system, distinguish constraints between those imposed by the starting 0-1 relaxation, and that are \textit{always} linear, and PSD constraints that are valid for \textit{all} 0-1 assignments (independently of the starting relaxation). As a result, any IG lower bound for strong LP relaxations that is based on a solution that comes from a global distribution of 0-1 assignments immediately translates into the same IG for a series of SDP hierarhies. A natural question that is raised is whether such global distributions of 0-1 assignments can be used to fool strong LP relaxations (and we answer this in the positive as we explain above). The second question that we raise is whether our solution is robust also against Lasserre tightenings. We answer this in the negative in Section~\ref{sec: las section high level}.  However, we prove two weaker integrality gaps in Sections~\ref{sec: l1 la} and~\ref{sec: general la} against the Lasserre system using a very similar construction to what we use for the \iSA\ system -- a $O(\sqrt{n})$ bound for the Level-$1$ Lasserre system and we give a superconstant integrality gap using a similar construction that holds for $\Theta(n)$ levels of the Lasserre system.

\cut{
, and we propose the all-$2^r\cdot t/\binom{n-2r}{2}$ vector.
}

%%%%%%%%%%%%%%%%%%%%%%%%%%%%%%%%%%%
%%%%%%%%%%%%%%%%%%%%%%%%%%%%%%%%%%%
%%%%%%%%%%%%%%%%%%%%%%%%%%%%%%%%%%%
%%%%%%%%%%%%%%%%%%%%%%%%%%%%%%%%%%%
\section{Preliminaries }\label{sec: preliminaries}

%\subsection{Notation}
We denote by $\bone_n$ the all-1 vector of dimension $n$, and we drop the subscript, whenever the dimension is clear from the context. Similarly, by all-$\alpha$ vector we mean the vector $\alpha \bone$. %By $I_n$ we denote the $n \times n$ identity matrix, and by $J_{n}$ the $n \times n$ all-1 matrix. 
For a fixed set of indices $[m]:=\{1,\ldots,m\}$, we denote by $\powerset_r$ all subsets of $[m]$ of size at most $r$ (for the partial vertex cover polytope and for a graph $\graph$, we will use $[m]=V \cup E$). For some $y \in \reals^{\powerset_{r+1}}$, we denote by $\mathcal Y$ the so-called \textit{moment matrix} of $y$ that is indexed by $\powerset_1$ in the rows and by $\powerset_r$ in the columns, with $\mathcal Y_{A,B} = y_{A\cup B}$. In other words, $\mathcal Y \in \reals^{|\powerset_1| \times |\powerset_r|}$ whenever $y \in \reals^{\powerset_{r+1}}$, whereas $\mathcal Y$ is a square symmetric matrix if $r=1$. Finally, we denote by $\{\be_I\}_{I \in \powerset_r}$ the standard orthonormal basis of $\powerset_r$, so that $\mathcal Y \be_A$ is the column of $\mathcal Y$ indexed by set $A$. 

\subsection{Problem Definition \& and a Natural LP Relaxation}\label{sec: problem def and LP relaxation}

Given an integer $t$, and a graph $\graph$\ with vertex weights $w_i \in \reals_+$ for each $i \in V$, \pvc\ can be alternatively defined as the following optimization problem where variables $\{x_q\}_{q \in V\cup E}$ are further restricted to be integral. 
\begin{align}
    \min  \quad & \sum_{i \in V} w_i ~ x_i & \tag{\pvc-LP} \label{equa: LP pvc} \\
    \textup{s.t.} \quad  & x_i + x_j \geq x_{e},& \quad \forall e=\{i,j\} \in E \label{equa: edge constraints} \\
					& \sum_{e \in E} x_e \geq t \label{equa: demand constraint}\\
					& 0 \leq x_q \leq 1& \quad \forall q \in V \cup E \label{equa: box constraints}
%{equa: edge constraints}{equa: demand constraint}{equa: box constraints}
\end{align}
Below we focus on uniform instances, in which $w_i=1$, for all $i \in V$. We denote the set of feasible solutions of the above LP as $P_t(G)$, or much simpler as $P_t$ when the underlying graph is clear from the context, and we call it the $t$-partial vertex-cover polytope. For each edge $e$, the reader should understand $x_e$ as the 0-1 indicator variable that says whether $e$ will be among the (at least) $t$ many that will be covered by some vertex, while for each vertex $i$, the 0-1 variables $x_i$ indicate whether vertex $i$ is chosen in the solution.

 \eqref{equa: LP pvc} is the starting point for the 2-approximation algorithm for \pvc\ in \cite{BB98}, and a $2-\Theta(1/d)$ approximation for unweighted instances, where $d$ the maximum degree of the input graph, in \cite{GKS04,Sri01}. Strictly speaking, the analysis that guarantees the 2-approximability is not relative to the performance of the LP for all instances, as in fact \eqref{equa: LP pvc} has an unbounded integrality gap.
\begin{observation}[Star-graph fools \eqref{equa: LP pvc}~\cite{Mes09}]\label{obs: star and LP}
Consider the unweighted star-graph $\graph$\ with $V={1,\ldots,n,n+1}$, and edges $\{n+1,i\} \in E$, for $i=1,\ldots,n$. The optimal solution to \pvc\ is 1, for every $t \in \naturals$. In contrast, consider the feasible solution to~\eqref{equa: LP pvc} that sets $x_e=x_{n+1}=t/n$ for all $e \in E$, and the rest of variables equal to 0. This gives a solution of cost $t/n$, hence the integrality gap of \eqref{equa: LP pvc} is at least $n/t$.
%, which in particular is $\Theta(n)$ for any constant $t$. 
\end{observation}
%Observation~\ref{obs: star and LP} indicates that no \eqref{equa: LP pvc}-based algorithm for \pvc\ can achieve a good approximation, whose analysis is relative to the performance of the LP. Even worse, \eqref{equa: LP pvc} cannot be used to guess the value of the optimal solution to a \pvc\ within factor $n/t -\epsilon$. 
For the algorithmic paradigm of LP-based algorithms (with performance analysis relative to the value of the relaxation), Observation~\ref{obs: star and LP} teaches us that natural LP relaxations may fail dramatically on simple graph instances. The reader should contrast this to the tractability of \pvc\ when $t$ is a constant, or when the input graph is a tree~\cite{CS13} (as this is the case in Observation~\ref{obs: star and LP}). Interestingly, we prove that this is also the case for a strong SDP relaxation of \pvc\ that has given its best approximation guarantee known. 
%For the proof of the proposition below, along with the SDP relaxation of~\cite{HS02}, see Appendix~\ref{sec: Halperin's SDP IG}. 

\ignore{
\begin{proposition}\label{prop: star fools halperins sdp}
For all $t\leq n/2$, the SDP of \cite{HS02} has integrality gap at least $n/t$ when the input is the star-graph of Observation~\ref{obs: star and LP}.
\end{proposition}
}

Given a graph $\graph$, and an integer $t$, the SDP relaxation introduced by Helperin and Srinivasan \cite{HS02} for the unweighted \pvc\ problem reads as follows.
\begin{align}
    \min  \quad & \frac12 \sum_{i \in V} \left( 1+\bv_0 \cdot \bv_i\right) & \tag{\pvc-SDP} \label{equa: SDP pvc} \\
    \textup{s.t.} \quad  
& \bv_0 \cdot \bv_i+\bv_0 \cdot \bv_j-\bv_i \cdot \bv_j \leq 1  ,& \quad \forall \{i,j\} \in E \label{equa: upper bound inner products SDP} \\
&\bv_0 \cdot \bv_i+\bv_0 \cdot \bv_j+\bv_i \cdot \bv_j \geq -1  ,& \quad \forall \{i,j\} \in E \label{equa: lower bound inner products SDP} \\
& \sum_{\{i,j\}\in E} \left(3+\bv_0 \cdot \bv_i+\bv_0 \cdot \bv_j-\bv_i \cdot \bv_j\right) \geq 4 t  ,& \quad  \label{equa: sum of edges in SDP}  \\ % \forall e=\{i,j\} \in E \notag \\
& \bv_i \in\reals^{|V|}, \|\bv_i\|=1,& \quad \forall i \in V \cup \{0\}  \label{equa: unit ball SDP}
\end{align}
The reader can verify that when restricted on integral solutions $\bv_i \in\reals^{1}$, \eqref{equa: SDP pvc} finds the optimal $t$-partial vertex cover $\{j \in V: ~\bv_i=\bv_0\}$ (note that when the vectors are unit dimensional, then each vector is equal to $\bv_0$ or to $-\bv_0$).  At the same time, it is an easy exercise that \eqref{equa: SDP pvc} is at least as strong \eqref{equa: LP pvc}, still both relaxations are fooled by the same bad integrality gap instance, as we prove next. 

\begin{proposition}\label{prop: star fools halperins sdp}
%\label{prop: star graph and SDP}
For all $t\leq n/2$, the integrality gap of \eqref{equa: SDP pvc} is at least $n/t$. %$\max\{ 2-\epsilon, n/t\}$ for every $\epsilon>0$.
\end{proposition}

\begin{proof}
%When $t\leq n/2$, 
We show that the star-graph of Observation~\ref{obs: star and LP} gives an integrality gap of $n/t$. Indeed, consider the following SDP vector solution in $\reals^2$: $\bv_0=-\bv_i=(1,0)$ for $i=1,\ldots,n$ and 
$$
\bv_{n+1} = \left(-1+2t/n, \sqrt{4t/n - 4t^2/n^2} \right).
$$
We examine now all constraints of~\eqref{equa: SDP pvc}. For every edge $\{i,n+1\}$ we have
\begin{align*}
\bv_0 \cdot \bv_i+\bv_0 \cdot \bv_{n+1}-\bv_i \cdot \bv_{n+1} 
-1 + (-1+2t/n) + (-1+2t/n) &= -3 + 4t/n \stackrel{t \leq n/2}{\leq} 1 \\
\bv_0 \cdot \bv_i+\bv_0 \cdot \bv_{n+1}+\bv_i \cdot \bv_{n+1} 
-1 + (-1+2t/n) - (-1+2t/n) &= -1 
\end{align*}
showing that \eqref{equa: upper bound inner products SDP} and \eqref{equa: lower bound inner products SDP} are satisfied. Next we check constraint ~\eqref{equa: sum of edges in SDP}
$$
\sum_{i=1}^n \left(3+\bv_0 \cdot \bv_i+\bv_0 \cdot \bv_{n+1}-\bv_i \cdot \bv_{n+1} \right) = n \cdot 4t/n = 4t. 
$$
Finally it is easy to see that all vectors above are unit, and that the value of the objective is indeed $t/n$, as required for an integrality gap of $n/t$. 
\end{proof}

We need to clarify that the statement of Proposition~\ref{prop: star fools halperins sdp} does contradict the fact that the best algorithm known for \pvc\ is based on~\eqref{equa: SDP pvc} and has performance strictly better than (but asymptotically equal to) 2. That should be of no surprise, since the approximation ratio achieved in~\cite{HS02} is due to an analysis relative to the performance of~\eqref{equa: SDP pvc} only for solutions of asymptotically large values. In particular, if $opt, sdp$ are the costs of the exact optimal solution and the optimal solution to~\eqref{equa: SDP pvc} respectively, the algorithmic analysis in~\cite{HS02} only relies on the highly non trivial relation 
$$opt\leq2\cdot sdp + 2$$
which allows for a large integrality gap, as indicated by Proposition~\ref{prop: star fools halperins sdp}. 

It is possible to show stronger integrality gaps for~\eqref{equa: SDP pvc}, especially when $t\geq n/2$, but this deviates from the subject of this work. The reader should keep that \eqref{equa: SDP pvc}, on which the best algorithm known for \pvc\ relies, cannot witness that a graph instance has a bounded solution, even with multiplicative error $n/t$, when $t\leq n/2$. That includes instances of \pvc\ that are tractable. Even more interestingly, and as we show in this work, a simple and natural linear program for which we prove strong negative results can solve the star-graph exactly. This constitutes a very unusual example of a specific instance of a combinatorial optimization problem for which a natural linear program outperforms (even in a single instance) an SDP that has been used in the best algorithm known for the same problem.

%%%%%%%%%%%%%%%%%%%%%%%%%%%%%%%%%%%
%%%%%%%%%%%%%%%%%%%%%%%%%%%%%%%%%%%
%%%%%%%%%%%%%%%%%%%%%%%%%%%%%%%%%%%
%%%%%%%%%%%%%%%%%%%%%%%%%%%%%%%%%%%

%\subsection{The \iLSp\ and \iSA\ Lift-and-Project Systems}
\subsection{Hierarchies of LP and SDP relaxations}\label{sec: Hierarchies of LP and SDP relaxations}

In this section we introduce families of LPs and SDPs derived by the so-called \iLS, \iLSp~\cite{LS91}, \iSA~\cite{SA90} and \iSAp\ systems. Starting with a polytope $P \subseteq [0,1]^m$, each of the systems derives a nested sequence of relaxations $\{P^{(r)}\}_{r=1,\ldots,m}$, such that $P^{(m)}=conv\left(P\cap\{0,1\}^m\right)$, while under mild assumptions one can optimize over $P^{(r)}$ in time $m^{O(r)}$. For an instance $\graph$\ of \pvc, our intention is to derive and study this sequence of relaxations starting with $P=P_t(G)$, i.e. the feasible region of the standard LP relaxation~\eqref{equa: LP pvc}, hence setting $|m|=|V| +|E|$. For the sake of simplicity, we adopt a unified exposition of the systems (see \cite{Lau03} for a more abstract exposition of lift-and-project systems).

For technical reasons, it is convenient to apply a standard homogenization to polytope $P$ as follows: variables $x_p$ are replaced by $\barr x_{\{p\}}$ and each constraint $a^Tx\geq b$ is replaced by  $a^T\barr x\geq b \barr x_\emptyset$. Adding the constraint $\barr x_\emptyset\geq 0$ along with the previous constraints define a cone that we denote by $K$. Clearly $K\cap\{\barr x_\emptyset =1\}$ is exactly polytope $P$. 
Next we define a sequence of LP and SDP refinements of an arbitrary 0-1 polytope, proposed by Lov\'asz and Schrijver~\cite{LS91}, and that is commonly known in the literature as the \iLS\ and \iLSp\ hierarchies. 

\begin{definition}[The \iLS\ system]~\label{def: LS}
Let $K^{(0)}:=K$ be a conified polytope $P \subseteq [0,1]^m$. The level-$r$ \iLS~tightening of $K^{(0)}$ is defined as the cone
$$
K^{(r)} = 
\left\{
x \in \reals^{\powerset_1}:~ \exists y \in \reals^{\powerset_2}
~\textrm{such that}~~
\begin{array}{ll}
\mathcal Y\be_\emptyset=x~\textrm{ and} \\
\forall i \in [m], ~\mathcal Y\be_{\{i\}}, \mathcal Y\left( \be_\emptyset - \be_{\{i\}} \right) \in K^{(r-1)}
\end{array}
\right\}
$$ 
The level-$r$ \iLS\ tightening $\mathcal N^{(r)}(P)$ of $P$ is obtained by projecting $K^{(r)}$ onto $x_\emptyset = 1$, i.e. $\mathcal N^{(r)}(P) = K^{(r)} \cap \{ x \in \reals^{\powerset_1}:~ x_\emptyset =1\} $. 
\end{definition}

\begin{definition}[The \iLSp\ system]~\label{def: LS+}
Let $K^{(0)}:=K$ be a conified polytope $P \subseteq [0,1]^m$. The level-$r$ \iLSp~tightening of $K^{(0)}$ is defined as the cone
$$
K^{(r)}_+ = 
\left\{
x \in \reals^{\powerset_1}:~ \exists y \in \reals^{\powerset_2}
~\textrm{such that}~~
\begin{array}{ll}
\mathcal Y \psd \bzero, ~~\mathcal Y\be_\emptyset=x~\textrm{ and} \\
\forall i \in [m], ~\mathcal Y\be_{\{i\}}, \mathcal Y\left( \be_\emptyset - \be_{\{i\}} \right) \in K^{(r-1)}_+
\end{array}
\right\}
$$ 
The level-$r$ \iLSp\ tightening $\mathcal N^{(r)}_+(P)$ of $P$ is obtained by projecting $K^{(r)}_+$ onto $x_\emptyset = 1$. 
\end{definition}
%%%%recover for camera ready
%For the intuition of the \iLSp\ system see Appendix~\ref{sec: intuition of ls lsp systems}.
%%%%end of recover for camera ready

%Note that the \iLS\ system is simply as \iLSp\ with the PSD constraint removed. 
%%%%remove for camera ready
The intuition of the technical Definition~\ref{def: LS+} is simple, at least for the level-1 relaxation; multiply each constraint of polytope $P$ by degree-1 polynomials $x_i, 1-x_i$ (for all $i$), and after expanding the quadratic expressions, substitute $x_i \cdot x_j$ by a brand new linear variable $y_{\{i,j\}}$, effectively simulating the identity $x_i^2=x_i$ which is valid in $P\cap \{0,1\}^m$. For example asking that $\mathcal Y\left( \be_\emptyset - \be_{\{i\}} \right) \in K^{(0)}$ is the same as multiplying all constraints of $P$ by $1-x_i$, and after linearizing as described above, and asking that the linear system is feasible. Therefore, the vectors $y \in \reals^{\powerset_2}$ of Definition~\ref{def: LS+} are meant to simulate monomials of degree at most 2, whereas the corresponding moment matrix $\mathcal Y$ for integral solutions $\barr x$ is simply the rank 1 positive definite matrix $\barr x \barr x^T$, hence the valid constraint $\mathcal Y\psd \bzero$. 
%%%%end of remove for camera ready

Next we introduce the \iSA\ system defined by Sherali and Adams~\cite{SA90} that derives a sequence of LP relaxations (and not SDP relaxations). 
%%%recover for camera ready
%For some intuition see Appendix~\ref{sec: intuition of sa system}. 
%%% of of recover for camera ready
\begin{definition}[The \iSA\ system]~\label{def: SA}
Let $K$ be a conified polytope $P \subseteq [0,1]^m$. The level-$r$ \iSA\ tightening of $K$ is defined as the cone
$$
M^{(r)} = 
\left\{
x \in \reals^{\powerset_1}:~ \exists y \in \reals^{\powerset_{r+1}}
~\textrm{such that}~~
\begin{array}{ll}
\mathcal Y\be_\emptyset=x\textrm{, and} \\
\forall Y,N~\textrm{wtih}~Y\cup N \in \powerset_r, %\forall B \subseteq A,
~\mathcal Y \sum_{\emptyset \subseteq T \subseteq N }(-1)^{|T|}\be_{Y\cup T}\in K
\end{array}
\right\}
$$ 
The level-$r$ \iSA\ refinement (tightening) $\mathcal S^{(r)}(P)$ of $P$ is obtained by projecting $M^{(r)}$ onto $x_\emptyset = 1$, i.e. $\mathcal S^{(r)}(P) = M^{(r)} \cap \{ x \in \reals^{\powerset_1}:~ x_\emptyset =1\} $.
%\footnote{Occasionally we abuse notation and we treat $\mathcal S^{(r)}(P)$ as a subset of $[0,1]^m$, instead of $\{x \in [0,1]^{m+1}: x_\emptyset =1\}$.}
\end{definition}
Occasionally we abuse notation and we treat $\mathcal N^{(r)}_+(P), \mathcal S^{(r)}(P)$ as subsets of $[0,1]^m$, instead of $\{x \in [0,1]^{m+1}: x_\emptyset =1\}$. Also,  relaxations derived by \iLSp\ and \iSA\ are in principle incomparable. 

The intuition behind the technical Definition~\ref{def: SA} is as follows; multiply each constraint of polytope $P$ by dergee-$r$ polynomials of the form $\prod_{i \in Y}x_i \prod_{j \in N}(1-x_j)$, for some sets $Y \cup N \in \powerset_r$. After expanding the high degree polynomial expressions, substitute $\prod_{i \in A} x_i$ by a brand new linear variable $y_{A}$, effectively simulating the identity $x_i^k=x_i$ for all $k=1,\ldots,r+1$, which is valid constraint in $P\cap \{0,1\}^m$. For example note that by expanding and linearizing $\prod_{i \in Y} x_i\prod_{j \in  N}(1-x_j)$ we obtain $\sum_{\emptyset \subseteq T \subseteq N }(-1)^{|T|}y_{Y\cup T}$, hence the seemingly complicated sum in the definition of $M^{(r)}$ above. 

For the reader familiar with \lap\ systems, it is easy to see that level-1 \iSA\ tightening coincides with the so-called level-1 \ls-LP tightening.% (that would be $\mathcal N^{(1)}_+(P)$ without the PSD constraint). 
Next we show that this seemingly weak LP solves the star graph. %For the proof see Appendix~\ref{sec: proof of ls level 1 solves star}.

\begin{proposition}\label{prop: LS1 solves star}
Let $G$ be the star graph of Observation~\ref{obs: star and LP}. Then the level-1 \iSA\ %(hence \ls\ as well) 
tightening of $P_t(G)$ has integrality gap 1. 
\end{proposition}

\begin{proof}
Let $x$ be a vector in the level-1 \iSA\ tightening of $P_t(G)$, and let $y$ be its moment matrix $\mathcal Y$ as in Definition~\ref{def: SA}. Let $\bb, \barr{\bb}, \bd, \barr{\bd} \in \reals^n$ and $a\in \reals$ be such that
$\mathcal Y\be_\emptyset = \left( 1, \bb^T, a, \bd^T \right)^T$ and 
$\mathcal Y\be_{\{n+1\}} = \left(a, {\barr{\bb}}^T, a, {\barr{\bd}}^T \right)^T$, where we explicitly assume that the list of indices has first all vertices (with the center being last), followed by all edges. Note that with this terminology, the value of the objective for such a solution is $a+ \bone^T_n \bb$, which we need to compare to $opt=1$. 

Next we focus on $\mathcal Y(\be_\emptyset - \be_{\{n+1\}})$ that satisfies all homogenized constraints of $P_t(G)$, and in particular constraints~\eqref{equa: edge constraints}
%{equa: edge constraints}{equa: demand constraint}{equa: box constraints}
of edges $\{n+1,i\}$, $i=1,\ldots,n$, which require that $\bb-\barr \bb \geq \bd-\barr \bd$. Similarly, constraint~\eqref{equa: demand constraint} of $P_t(G)$ 
%(that requires $t$ many edges to be covered) 
implies that $\bone^T_n (\bd-\barr \bd) \geq (1-a)t$. Therefore
$$
a+ \bone^T_n \bb \geq a + \bone_n^T (\bd-\barr \bd) \geq a + (1-a)t \geq 1 = opt.
$$
%exactly as wanted. 
\end{proof}

An alternative proof of  Proposition~\ref{prop: LS1 solves star} follows by using the conditioning property according to which any vector solution in level-1 \iSA\ projected space is a convex combination of feasible solutions to the original LP that are integral in any index. Choosing as index the center of the star implies that any feasible solution to level-1 \iSA\ is a convex combination of a solution with cost 1 and a solution with cost $t$. 

Recall that by Proposition~\ref{prop: star fools halperins sdp} the star graph is also responsible for a $n/t$ integrality gap for the SDP of~\cite{HS02}, i.e. the relaxation which the best algorithm known for \pvc\ is based on. The surprising conclusion from Proposition~\ref{prop: LS1 solves star} is that a simple LP that one can derive systematically from $P_t(G)$ outperforms that particular SDP for a specific instance. This is in contrast to other known examples of level-$\Theta(m)$ \iLS\ tightenings that are strictly weaker than natural and static SDP relaxations. 
Finally, it is worthwhile mentioning that we do not know whether constant-level \lap\ tightenings of \eqref{equa: LP pvc} derive the SDP of~\cite{HS02}.

\cut{
, e.g. for \textsc{Max-Cut} (see~\cite{GW95,STT07b}). 
}

For algorithmic purposes, a number of \iSA\ variants have been proposed that give rise to hierarchies of SDPs (see \cite{AT13} for a list of them). The simplest variation, and the one that has resulted surprisingly strong positive results, is usually referred as the mixed hierarchy. This system, that we denote here by \iSAp\ imposes an additional PSD constraint. 

\begin{definition}[The \iSAp\ system]~\label{def: SAp}
Let $K$ be a conified polytope $P \subseteq [0,1]^m$. The level-$r$ \iSAp\ tightening of $K$ is defined as the refinement of cone $M^{(r)}$, as in Definition~\ref{def: SA}, where the $(m+1)$-leading principal minor of the moment matrix $\mathcal Y$, i.e. the principal minor of $\mathcal Y$ that is indexed by sets of variables of size at most 1, is PSD. 
\end{definition}

Next we give a formal definition of the \iLa\ system, which is a refinement of the \iSAp\ system. 

\begin{definition}[The \iLa\ system]~\label{def: La}
Consider some polytope $P \subseteq [0,1]$. 
The so-called level-$r$ \iLa\ SDP is a collection of PSD constraints on vectors $y$ indexed by $\powerset_{2r+2}$.
For each $y$, its \iLa-moment matrix $\mathcal Z$ is indexed in the rows and in the columns by $\powerset_{r+1}$, such that $\mathcal Z_{A,B}=y_{A\cup B}$. For each constraint $\sum_i \alpha_i^{(l)}x_i - \beta^{(l)} \geq 0$ of $P$, its slack moment matrix $\mathcal Z^{(l)}$ is indexed by $\powerset_{r}$, such that $\mathcal Z^{(l)}_{A,B}=\sum_i \alpha_i^{(l)} y_{A\cup B \cup \{i\}} - \beta^{(l)} y_{A\cup B}$. Then the level-$r$ \iLa\ SDP requires that all  matrices $\mathcal Z$ and $\{\mathcal Z^{(l)}\}_l$ are PSD (constraints that are valid for the integral hull of $P$).
The level-$r$ \iLa\ tightening of $P$ is the collection of projected vectors $y$ onto indices in $\powerset_{1}$.
\end{definition}
 Notably, the PSDness of proper principal minors of matrices $\mathcal Z$ and $\{\mathcal Z^{(l)}\}_l$ in the above definition is equivalent to the level-$r$ \iSA\ linear constraints~\cite{Lau03}. As such, the level-$r$ \iLa\ SDP is at least as strong as the level-$r$ \iSA\ LP. 
Level-$r$ SDPs derived by the \iSAp\ and \iLSp\ systems are not comparable. 
%In Section~\ref{sec: lower bounds for various SDPs} we introduce a further refinement of \iSAp\ that is strictly tighter than \iLSp, and for which we actually derive the same IG lower bounds as in \iSA. We postpone its definition due to its technicality.

By the generic algorithmic properties common to \iLSp, \iSA, \iSAp\ and \iLa\ systems, and for the \pvc\ polytope, it is immediate that for any graph $\graph$\, the level-($|V|+|E|$) relaxations have integrality gap 1. However, from the proof of convergence from all systems, it easily follows that vectors in level-$r$ relaxations satisfy \textit{any} constraint that is valid for the integral hull of $P_t(G)$ and that has support at most $r$. If $opt$ denotes the optimal value for $\graph$\, then $\sum_{i \in V} x_i \geq opt$ is a constraint valid for every integral solution with support $|V|$. Hence, level-$|V|$ LPs or SDPs derived by \iSA, \iLSp\ and \iSAp\ systems can solve any \pvc\ instance exactly. Can level-$r$ relaxations close the unbounded inegrality gap of $P_t(G)$ as exhibited in Observation~\ref{obs: star and LP}, for $r=o(|V|)$? We answer this question in the negative in the next sections by proving strong integrality gaps for superconstant level LP and SDP relaxations. As a byproduct, we show this way that LPs and SDPs that give rise to algorithms that run in superpolynomial time cannot solve to any good proximity even the tractable combinatorial problem \pvc\, where $t = \Theta(1)$.

%%%%%%%%%%%%%%%%%%%%%%%%%%%%%%%%%%%
%%%%%%%%%%%%%%%%%%%%%%%%%%%%%%%%%%%
%%%%%%%%%%%%%%%%%%%%%%%%%%%%%%%%%%%
%%%%%%%%%%%%%%%%%%%%%%%%%%%%%%%%%%%

%\section{The \iSA\ lower bound}\label{sec: SA lower bound}
\section{IG lower bounds for the \sa\ LP system}\label{sec: SA lower bound}
This section is devoted in proving one of our main results. 
\begin{theorem}\label{thm: sa ig lower bound with parameters}
Let $n,r,t$ be integers with $n \geq 2r+2t+2$. Then the integrality gap of the level-$r$ \iSA-tightening of~\eqref{equa: LP pvc} on graphs with $n$ vertices is at least $\binom{n-2r}{2}/t\cdot n$. 
\end{theorem}
For this we fix a clique $\graph$ on $n$ vertices, along with $r,t$ such that $n \geq 2r+ 2t+2$. 
We start by presenting Random Process~\ref{def: distribution}, that defines a distribution of 0-1 assignments for variables of the polytope $P_t(G)$.

\begin{algorithm}
\caption{(Definition of distribution $\distr_p$)}
\label{def: distribution}
\begin{algorithmic}[1]
\REQUIRE A fixed $p \in [0,1]$.
\FOR{$i \in V$}
\STATE Independently at random, set $x_i=1$ with probability $p$
\ENDFOR
\FOR{$e \in  E$}
\STATE Set $x_e$ equal to 1 as long as $e$ is incident to some $i$ for which $x_i=1$, and otherwise to 0. 
\ENDFOR
\ENSURE Distribution $\distr_p$ induced by the experiment above.
%\RETURN $y_1, \ldots, y_n$
\end{algorithmic}
\end{algorithm}
%Before we proceed with some crucial observations, note that the subclique of step 1 of Random Process~\ref{def: distribution} exists since $n-2r\geq 4$ and since each element of $A$ is either a vertex or it touches 2 vertices. 

We are ready to propose a vector solution $y\in \reals^{\powerset_{r+1}}$ to the level-$r$ \iSA\ tightening of $P_t(G)$. For $A \in \powerset_{r+1}$ (with ground set $V\cup E$), and for each $q \in A$, let $X_q$ be the random variable which equals 1 if $x_q=1$ in the random experiment of $\distr_{p}$, and 0 otherwise. For all such $A \subseteq V\cup E$, we define 
\begin{equation}\label{equa: def of SA solution}
y_A := \Exp{\distr_{p}}{\prod_{q \in A} X_q} = \pr{\distr_{p}}{\forall q \in A, ~x_q = 1} 
\end{equation}
where the last equality is due to that $X_q$ are 0-1 variables. In particular, this means that for all $i \in V$ and $f \in E$ we have
\begin{equation}\label{equa: singleton SA values}
y_{\{i\}} = p, ~~y_{\{f\}} = 2p-p^2,
\end{equation}
where $2p-p^2$ is the probability that at least one endpoint of edge $f$ is chosen, minus the probability that both are chosen (i.e the probability that edge $f$ is covered). The following is a standard observation that is used in many \iSA\ lower bounds, and that makes explicit the probabilistic interpretation of the system. 
%crucial observation that %uses Lemma~\ref{lem: local consistency} and 
%makes explicit the probabilistic interpretation of the \iSA\ system. 
%%%%recover for camera ready
%%Its proof can be found in Appendix~\ref{sec: local probabilities}. 
%%%% end of recover for camera ready
\begin{lemma}\label{lem: local probabilities}
For $Y \cup N \in \powerset_{r+1}$, let $w_{Y,N} := \sum_{\emptyset \subseteq T \subseteq N }(-1)^{|T|} y_{Y \cup T}$. Then  \\
$
w_{Y,N} = \pr{\distr_{p}(Y \cup N)}{ \forall q \in Y, X_q=1,~\&~\forall q' \in N, X_{q'}=0}.
$
\end{lemma}
\begin{proof}
\begin{align}
\sum_{\emptyset \subseteq T \subseteq N }(-1)^{|T|} y_{Y\cup T}
& 
=
\sum_{\emptyset \subseteq T \subseteq N }(-1)^{|T|} \Exp{\distr_{p}}{\prod_{q \in Y\cup T} X_q} \notag \\
&=
\Exp{\distr_p}{\sum_{\emptyset \subseteq T \subseteq N }(-1)^{|T|} \prod_{q \in Y\cup T} X_q} \tag{Linearity of expectation}\\
&=
\Exp{\distr_p}{\prod_{q \in Y} X_q \prod_{q' \in N} \left( 1- X_{q'} \right) } \tag{$X_q \in \{0,1\}$}\\
& =
\pr{\distr_p}{ \forall q \in Y, x_q=1,~\&~\forall q' \in  N, x_{q'}=0}
\notag
\end{align}
\end{proof}

We can now prove that $y$ is solution to the level-$r$ \iSA\ polytope of \pvc, for a proper choice of $p$. 
\begin{lemma}\label{lem: main SA pvc}
For the complete graph $\graph$\ on $n$ vertices, and for all $r, t$ with $n \geq 2r+2t+2$, let $y\in \reals^{\powerset_{r+1}}$ be as in~\eqref{equa: def of SA solution}, where $p=t/\binom{n-2r}{2}$. Then $y \in S^{(r)}(P_t(G))$. 
\end{lemma}

\begin{proof}
Let $Y,N \in \powerset_r$ with $|Y \cup N| \leq t$. We need to show that $\barr y:=\mathcal Y \sum_{\emptyset \subseteq T \subseteq N }(-1)^{|T|}\be_{Y\cup T} \in \reals^{\powerset_1}$ satisfies all constraints of $P_t(G)$ (after they are homogenized). 

Asking that $\barr y$ satisfies the constraint~\eqref{equa: edge constraints} for an edge $e=\{i,j\}$ %, i.e. $x_i + x+j \geq z_e$ 
is the same as asking that $w_{Y \cup \{i\},N} + w_{Y \cup \{j\},N} -w_{Y \cup \{e\},N}  \geq 0$. Note that $|Y\cup N \cup \{i,j\}| \leq r+2$.
%, so Random Process~\ref{def: distribution} does define $\distr_{p}(Y\cup N \cup \{i,j\})$.
Due to Lemma~\eqref{lem: local probabilities} and by linearity of expectation we have
$$
w_{Y \cup \{i\},N} + w_{Y \cup \{j\},N} -w_{Y \cup \{e\},N} 
=
\Exp{\distr_{p}(Y \cup N \cup \{i,j\} )}
{\prod_{q \in Y} X_q \prod_{p \in N} \left( 1- X_p \right) (X_i + X_j - X_e )}.
$$
But recall that in Random Process~\ref{def: distribution} we we set $x_e=1$ only when at least one among $x_i,x_j$ is already set to 1. Therefore the previous expected value is always non negative. 

In a similar manner we can show that box constraints~\eqref{equa: box constraints} are satisfied. First, constraints of the form $x_q \geq 0$, $q \in V \cup E$ are satisfied for $\barr y$, since by Lemma~\ref{lem: local probabilities},  $w_{Y\cup \{q\},N}$ represents a probability of an event. As for constraints $x_q\leq 1$, we need to prove that $w_{Y\cup \{q\}, N} \leq w_{Y, N}$. This is true again due to Lemma~\ref{lem: local probabilities}, and because the event associated with $w_{Y,N}$ is logically implied by that of $w_{Y \cup \{q\}, N}$. 

Finally we need to show that $\barr y$ satisfies constraint~\eqref{equa: demand constraint}, i.e. constraint $\sum_{e \in E} w_{Y\cup \{e\},N} \geq t \cdot w_{Y,N}$. For this we recall that $|Y \cup N | \leq r$, and so in the original clique on $n$ vertices, there is a subclique $G'=(U,F)$ on at least $n-2r\geq 4$ vertices, such that no edge in $F$ is incident to any element (vertex or edge) in $Y\cup N$, and $|F| \geq \binom{n-2r}{2}>0$. This means that for every $f \in F$ the event that $X_f=1$ is independent to any 0-1 assignment on variables in $Y \cup N$, while $\pr{\distr_p}{X_f=1} \stackrel{\eqref{equa: singleton SA values}}= 2p-p^2 \geq p$, since $p=t/\binom{n-2r}{2}\leq t/\binom{2t+2}{2}<1/2$. Since we also have $|F| \cdot p = |F| \cdot t/\binom{n-2r}{2} \geq t$, we conclude that 
$
\sum_{e \in E} w_{Y\cup \{e\},N} 
\geq 
\sum_{e \in F} w_{Y\cup \{e\},N} 
=
|F| \cdot p \cdot   w_{Y,N}
\geq
t\cdot w_{Y,N}
$, 
as promised.
\end{proof}
Note that by~\eqref{equa: singleton SA values}, and for the value of $p$ as in Lemma~\ref{lem: main SA pvc}, the objective of the level-$r$ \iSA\ LP is no more than $n \cdot p = t\cdot n / \binom{n-2r}{2}$, while the optimal solution of the input graph has cost 1, concluding the proof of Theorem~\ref{thm: sa ig lower bound with parameters}.

It is worthwhile noticing that our superconstant integrality gaps lower bounds hold only for values of parameter $t =o(n)$. The reader can easily verify that when the input is the $n$-clique, then the optimal solution to~\eqref{equa: LP pvc} is exactly $t/(n-1)$ (e.g. using the dual of~\eqref{equa: LP pvc}). Therefore, for any constant $c$ and when $n/c \leq t\leq n-1$, for which the optimal solution to \pvc\ is still 1, the integrality gap of~\eqref{equa: LP pvc} is strictly less than $c$. In particular, the integrality gap drops below 2 when $c\geq 2$. 

%%%%%%%%%%%%%%%%%%%%%%%%%%%%%%%%%%%
%%%%%%%%%%%%%%%%%%%%%%%%%%%%%%%%%%%
%%%%%%%%%%%%%%%%%%%%%%%%%%%%%%%%%%%
%%%%%%%%%%%%%%%%%%%%%%%%%%%%%%%%%%%

\section{IG lower bounds for various SDP hierarchies}\label{sec: lower bounds for various SDPs}

\subsection{SDPs derived by the \iSAp\ and \iLSp\ systems}\label{sec: if for lsp and sap}

In this section we argue that the moment matrix $\mathcal Y$ of solution $y$ that we proposed in Lemma~\ref{lem: main SA pvc} satisfies very strong PSD conditions. This will immediately imply the same IG lower bounds of Theorem~\ref{thm: sa ig lower bound with parameters} also for stronger SDP systems, as summarized in the next theorem. 
\begin{theorem}\label{thm: lsp sap ig lower bound with parameters}
Let $n,r,t$ be integers with $n \geq 2r+2t+2$. Then the integrality gap of the level-$r$ \iLSp\ and \iSAp\ tightenings of~\eqref{equa: LP pvc} on graphs with $n$ vertices is at least $\binom{n-2r}{2}/t\cdot n$. 
\end{theorem}

For proving Theorem~\ref{thm: lsp sap ig lower bound with parameters}, we fix the clique $\graph$ on $n$ vertices, together with $r,t$ such that $n\geq 2r+2t+2$. In all our arguments below we use $y \in \reals^{\powerset_{r+1}}$ as defined in~\eqref{equa: def of SA solution}, as well as vector $w$ (indexed by pairs of sets of variables) as it appears in Lemma~\ref{lem: local probabilities}. We also define the matrix $\mathcal X^{Y,N} \in \reals^{\powerset_1 \times \powerset_1}$, which at entry $A,B$ (i.e. any two sets of size at most 1) equals $w_{Y\cup A \cup B, N}$. Note that matrix $\mathcal X^{Y,N}$ is exactly the moment matrix of random variables $\{X_q\}_{q\in V\cup E}$ condition on $X_q=1$ for all $q\in Y$, and $X_{q'}=0$ for all $q' \in N$, scaled by the constant $\pr{\distr_p}{\forall q \in Y, X_q=1~\&~\forall q'\in N, X_{q'}=0}$. In particular, for each $q \in V \cup E$ we have that vectors $\mathcal X^{Y,N}\be_q, \mathcal X^{Y,N} (\be_\emptyset - \be_q)$ satisfy all constraints of $P_t(G)$. 
We have the following well known observation. 

\begin{observation}\label{obs: psdness of minors}
Let $Y,N$ be any subsets of $V \cup E$ such that $|Y \cup N| \leq r-1$. Then $\mathcal X^{Y,N}$ is positive semidefinite. 
\end{observation}
Indeed, recall that $y \in \reals^{\powerset_{r+1}}$ is obtained by the global distribution $\distr_p$ that associates \textrm{any} 0-1 assignment of variables of $P_t(G)$ with some probability. In particular, if $x \in \{0,1\}^{\powerset_1}$, with $x_\emptyset =1$, is such a 0-1 assignment, then $xx^T$ is a rank 1 PSD matrix. Clearly, matrix $\mathcal X^{Y,N}$ is a convex combination of such rank-1 PSD matrices, hence it is PSD as well.

It is now immediate that our \iSA\ solution $y$ satisfies also the extra PSD constraint imposed by \iSAp. What we only need to observe is that the leading principal minor of $\mathcal Y$ indexed by sets of size at most 1 is exactly $\mathcal X^{\emptyset,\emptyset}$, which is PSD by Observation~\ref{obs: psdness of minors}. Hence, Theorem~\ref{thm: sa ig lower bound with parameters} also holds when \iSA\ tightenings are replaced by \iSAp\ tightenings. 

Next we argue that our \iSA\ solution is robust against much stronger SDP refinements. Note that vector $w$ is well defined for all level-$r$ \iSA\ solutions $y$. Especially when $y$ is obtained as a convex combination of integral vectors, all matrices $\mathcal X^{Y,N}$ are PSD, for all $|Y\cup N|\leq r-1$. That is, the latter constraints constitute a further refinement of the \iSAp\ system. % that is commonly used for algorithmic purposes.
Again by Observation~\ref{obs: psdness of minors} it is immediate that our level-$r$ \iSA\ solution fools also these exponentially many (in $r$) PSD conditions. What makes this new observation interesting is that these new PSD refinements are stronger than the constraints derived by the level-$(r-1)$ \iLSp\ system (see~\cite{Tou06p}). At a high level, this is true due to an alternative inductive definition of the \iSA\ system (similar to the inductive definition of the \iLSp\ system) that allows to use matrices $\mathcal X^{Y,N}$ as the ``protection moment matrices" required by Definition~\ref{def: LS+}.

\subsection{On SDPs derived by the \la\ system}\label{sec: las section high level}

In light of the discussion in Section~\ref{sec: if for lsp and sap}, a natural question to ask is whether our \iSA\ solution fools SDPs derived by the so-called Lasserre (\iLa) system~\cite{Las01}. For completeness, we briefly elaborate on this question, by concluding that the level-1 SDP derived by the \iLa\ system does eliminate our bad integrality gap solution. For convenience we consider $P=P_t(G)$ as the underlying polytope that is to be tightened.

Our proposed \iSA\ solution Lemma~\ref{lem: main SA pvc} can be easily seen to satisfy many level-$r$ \iLa\ PSD-constraints but one. In fact, we can show that even the level-1 \iLa\ SDP is not fooled by our \iSA\ solution. 

\begin{lemma}\label{lem: no las sol}
For any constant $r$, the level-1 \iLa\ SDP eliminates the level-$r$ solution proposed in Lemma~\ref{lem: main SA pvc}.
\end{lemma}

\begin{proof}
Fix $n,t,p$, and let $y$ be the solution to the level-$(r)$ \iSA-tightening as described in Lemma~\ref{lem: main SA pvc}. Recall that our \pvc\ instance is the complete graph $\graph$ on $n$ vertices, in which every vertex is chosen independently at random with probability $p$. 

For completeness, first we briefly elaborate on \iLa\ PSD constraints that are satisfied. 
Moment matrix $\mathcal Z$ along with the slack matrices of constraints \eqref{equa: edge constraints}, \eqref{equa: box constraints} are all PSD (and this remains true even for level $\lfloor r/2 \rfloor$ \iLa\ PSD constraints). The argument for this is identical to the one used to prove Observation~\ref{obs: psdness of minors} (recall that $y$ is obtained from a global distribution of 0-1 assignments). 

It therefore remains to check the PSDness of the level-1 slack matrix of the demand constraint~\eqref{equa: demand constraint}. In order to prove that this matrix is not PSD, it suffices to focus on its principal minor ${\mathcal M}$ that is indexed only by subsets of vertices. To that end, let $y_A \in \reals^{\mathcal P_1}$ be the indicator vector of set $A\subseteq V$. Let also $L_n$ denote the expected slack we have in constraint~\eqref{equa: demand constraint} when each vertex is chosen with probability $p$ in the $n$-clique, and $C_{n,a}$ be the number of edges that are covered by choosing $a$ many vertices in the same graph. Then, it is easy to verify by definition that ${\mathcal M}$ has the form 
\begin{align}
%\barr{\mathcal Z} = \sum_{A\subseteq V} P_{|A|} S_{|A|} y_A y_A^T
\left({\mathcal M}\right)_{I,J} 
&
= \sum_{A\subseteq V} 
\underbrace{p^{|A|}(1-p)^{n-|A|}}_\text{Probability of choosing only vertices $A$}
\underbrace{
\left(
\binom{|A|}{2}+|A|(n-|A|)-t
\right)
}_\text{$C_{n,|A|}:=$ Slack of constraint~\eqref{equa: demand constraint} when choosing $A$}
\left( y_A y_A^T \right)_{I,J} \notag
\\
&= \sum_{I\cup J \subseteq A\subseteq V} p^{|A|}(1-p)^{n-|A|} \left(
\binom{|A|}{2}+|A|(n-|A|)-t
\right) \notag \\
&=
p^{|I\cup J|} 
\sum_{A\subseteq V\setminus \left( I\cup J \right)} 
p^{|A|}(1-p)^{(n-|I\cup J|)-|A|}
\left(
C_{n-|I\cup J|, |A|}-t + C_{n,|I\cup J|}
\right) \notag \\
&= 
p^{|I\cup J|} 
\left(
L_{n-|I\cup J|}
+
C_{n,|I\cup J|}
\right). \label{eq: value of slack matrix}
\end{align}
The second to last equality is obtained by factoring out $p^{|I\cup J|}$ and using the definition of $C_{n-|I\cup J|, |A|}$.
The last equality follows from the fact that $L_n=\sum_{a=0}^n\binom{n}{a}p^a(1-p)^{n-a}(C_{n,a}-t)$. 
%$P_a=p^{a}(1-p)^{n-a}$ is the probability of choosing only the set of vertices $A$ of size $a$, $S_A = \binom{a}{2}+a(n-a)-t$ is the slack of constraint~\eqref{equa: demand constraint} when the solution chooses $a$ many vertices, and 
Note also that $\left({\mathcal M}\right)_{\emptyset,\emptyset} = L_n = \binom{n}2 (2p-p^2) - t$. Applying the Schur complement on ${\mathcal M}$ with respect to the entry $\left({\mathcal M}\right)_{\emptyset,\emptyset}$, and given that $L_n>0$, we have that ${\mathcal M}$ is PSD if and only if 
$M - \frac{\left(
p^{} 
\left(L_{n-1}
+
C_{n,1}
\right)\right)^2}{L_n}J_n$
 is PSD, where $M$ is the minor of ${\mathcal M}$ indexed by sets of vertices of size 1, and $J_n$ is the all-one $n \times n$ matrix. By symmetry, all rows of $M$ have the same sum, i.e. the all-one vector $\bone$ is an eigenvector for the Schur complement. The corresponding eigenvalue can be computed by noticing that
\begin{align*}
& \left(
M - \frac{\left(
p^{} 
\left(L_{n-1}
+
C_{n,1}
\right)\right)^2}{L_n}J_n
\right) \bone \\
= &
\left(
p^{} 
\left(
L_{n-1}
+
C_{n,1}
\right)
+
(n-1)
p^{2} 
\left(
L_{n-2}
+
C_{n,2}
\right)
-
\frac{n\left(
p^{} 
\left(L_{n-1}
+
C_{n,1}
\right)\right)^2}{L_n}
\right) \bone
\end{align*}
Elementary calculations then show that the leading term of the eigenvalue above, when $p=c/n^2$, is $\left(-2 c^4-\frac{15 c^3}{2}-2 c^2\right)\frac1n<0$ (the rest of the summands are of order $o(1/n)$). 
\end{proof}
Interestingly, a slight modification of the proof of Lemma~\ref{lem: no las sol} can show that the solution proposed in Lemma~\ref{lem: main SA pvc} is violated by the level-1 \iLa\ SDP as long as $p=o\left( 1/n^{1.5} \right)$. 

%%%%%%%%%%%%%%%%%%%%%%%%%%%%%%%%%%%
%%%%%%%%%%%%%%%%%%%%%%%%%%%%%%%%%%%
%%%%%%%%%%%%%%%%%%%%%%%%%%%%%%%%%%%
%%%%%%%%%%%%%%%%%%%%%%%%%%%%%%%%%%%

\subsection{Lower bounds in the Level-$1$ Lasserre System}
\label{sec: l1 la}
However, the Lasserre system still admits a non-constant integrality gap at level $1$, so long as $p \geq O(t/n^{1.5})$.  Formally,
\begin{lemma}
There exists a solution to the Level-$1$ \iLa\ SDP which has integrality gap $O(\sqrt{n})$.
\end{lemma}

To construct this solution
we will need some further structural results.  We start by extending our construction of the vertex-indexed-minor only Lasserre slack matrix of constraint~\eqref{equa: demand constraint} given in the previous section to also account for edges.  For every $A \subseteq V[K_n]$, let $y_A$ be a vector indexed over subsets of $E[K_n] \cup V[K_n]$ defined in the following manner:
\[
(y_A)_B = \left\{\begin{array}{cl} 1 & \text{ if} B \subseteq (A \cup E[K_n]) \setminus E[K_n \setminus A] \\
	0 & \text{ otherwise.} \end{array}\right.
\]
It is not too hard to see that the construction for the Slack moment matrix also carries over.
\begin{lemma}
Let ${\mathcal Z}$ be the full Lasserre slack matrix for constraint~\eqref{equa: demand constraint} for the solution proposed in Lemma~\ref{lem: main SA pvc}.  Then
\[{\mathcal Z} =  \sum _{A \subseteq V} p^{|A|}(1-p)^{n - |A|}C_{n,|A|}y_A y_A^T\]
\end{lemma}
\begin{proof}
    Observe first that ${\mathcal Z} = \sum _{f \in E} Z_f - tZ$, for
    \[Z = \sum _{A \subseteq V} p^{|A|}(1-p)^{n - |A|} y_Ay_A^T\]
    \[Z_f = (2p - p^2)\sum _{A \subseteq V} \mathbb{P}_\mathbb{D}[A \to 1, A^c \to 0| f \to 1] y_A y_A^T \]

    From the definition of distribution $\mathbb{D}$ (note that the chosen vertices determine the edges) we have
    \[
      \mathbb{P}_\mathbb{D}[A \to 1, A^c \to 0 | f \to 1] = \left\{\begin{array}{cl}
\frac{p^{|A|}(1 - p)^{n - |A|}}{2p-p^2} & \mbox{if} f~\mbox{is incident to any vertex in}~A\\
         0 & \mbox{otherwise} %\mbox {if } f \subseteq A^c 
%         \frac{p^{|A|}(1 - p)^{n - |A|}}{2p-p^2} & \mbox {if } f \subseteq A \\
         \end{array}\right.
    \]

    Thus,
    \[
      Z_f = \sum _{A \subseteq V, |f \cap A| > 0} p^{|A|}(1 - p)^{n - |A|} y_A y_A^T.
    \]

    Therefore,
    \[
      \sum _{f \in E} Z_f = \sum _{A \subseteq V} p^{|A|}(1 - p)^{n - |A|}\left({|A| \choose 2} + |A|(n-|A|)\right) y_A y_A^T,
    \]
    as desired.
\end{proof}
Unfortunately ${\mathcal Z}$ is more complicated as it also refers to edge constraints.  However as it turns out
one only needs to worry about the vertex constraints in the Lasserre slack matrix.

\begin{lemma}
 \label{lem:edge-elim}
  Let $I, J \subseteq V[K_n] \cup E[K_n]$.  Let $I = F \cup D$, where $F \subseteq V$ and $D \subseteq E$.
  Consider the row $\mathcal{Z}_I$ of $\mathcal{Z}$.  Then $\mathcal{Z}_I$ can be written as a linear combination of the rows indexed
  by subsets of $F \cup (\bigcup _{e \in D} e)$.
\end{lemma}
\begin{proof}
  Proceed by induction on the size of $D$.  The proof is direct when $|D| = 0$.  Suppose $|D| = k$.
  Let $e = \{x, y\} \in D$.

  Then,
    \begin{align*}
      (\mathcal{Z}_I)_J &= \sum _{I \cup J \subseteq A \cup E[A], A \subseteq V}\ p^{|A|}(1-p)^{n - |A|}C_{n,|A|}\\
              &= \sum _{F \cup (D \setminus e) \cup \{x\} \cup J \subseteq A \cup E[A], A \subseteq V} p^{|A|}(1-p)^{n - |A|}C_{n,|A|}\\
              &+ \sum _{F \cup (D \setminus e) \cup \{y\} \cup J \subseteq A \cup E[A], A \subseteq V} p^{|A|}(1-p)^{n - |A|}C_{n,|A|}\\
              &- \sum _{F \cup (D \setminus e) \cup \{x, y\} \cup J\subseteq A \cup E[A], A \subseteq V} p^{|A|}(1-p)^{n - |A|}C_{n,|A|}
      \end{align*}

  Now, by the induction hypothesis all $3$ components can be written as a linear combination of rows indexed by subsets
  of $C \cup (\bigcup _{e \in D} e)$.  So $Z_I$ can be written as a linear combination of the rows
  indexed by subsets of $F \cup (\bigcup _{e \in D} e)$.
\end{proof}

Thus:
\begin{lemma}
  The level-$r$ Lasserre minor of $\mathcal{Z}$ (moment matrix of constraint~\eqref{equa: demand constraint}) is positive semidefinite if the level-$2r$ vertex-subset-only
  Lasserre minor of $\mathcal{Z}$ is positive semidefinite.
\end{lemma}
\begin{proof}
Let $M$ be the level-$r$ Lasserre minor of $\mathcal{Z}$.
Let $N$ be the minor of $\mathcal{Z}$ formed by taking all subsets $S$ of vertices and edges of $E$ such that $S$ contains or is incident to at most $2r$ vertices in $G$.  Clearly $M$ is a symmetric minor of $N$, so $N \succeq 0 \implies M \succeq 0$. Now, from Lemma~\ref{lem:edge-elim},
we can eliminate the rows $N_s$ indexed by subsets $S$ containing edges using rows indexed by
subsets containing up to $2r$ vertices.  Let $Q$ be the elementary matrix that encodes these elementary row operations.  Symmetrically, as $N$ is symmetric, $Q^T$ will eliminate the columns of $N$ indexed by subsets $S$ containing edges using columns indexed by
subsets containing up to $2r$ vertices.

Now 
$QNQ^T$ is exactly the level-$2r$ vertex subset only Lasserre minor of $\mathcal{Z}$.
So $QNQ^T \succeq 0 \iff N \succeq 0 \implies M \succeq 0$, as desired.
\end{proof}

Let $\Z$ denote the level-$2$ vertex-subset-only Lasserre minor of $\mathcal{Z}$. 
Let also 
$$S_k := p^k
\left(k(n-k) + {k \choose 2} + {n - k \choose 2}(2p - p^2) - t\right),$$
see also~\eqref{eq: value of slack matrix}. 
 It is not too hard to see that $\Z$ has the following form. 
\begin{small}
\[
 \Z =
\left[ \begin{array}{c|ccc}
	   &\varnothing & J, |J| = 1 & J, |J| = 2 \\ \hline
      \varnothing&S_0 & S_1 & S_2 \\
      I,|I|=1&S_1 &
      \left\{\begin{array}{cc}
          S_1 & \mbox{ if same } \\
          S_2 & \mbox{ if different } \end{array}\right\} & 
      \left\{\begin{array}{cc}
          S_2 & \mbox{ if intersect } \\
          S_3 & \mbox{ if not } \end{array}\right\} \\
      I,|I|=2&S_2 &
      \left\{\begin{array}{cc}
          S_2 & \mbox{ if intersect } \\
          S_3 & \mbox{ if not } \end{array}\right\} \ &
      \left\{\begin{array}{cc} 
          S_2 & \mbox { if equal } \\
          S_3 & \mbox { if 1 vertex shared } \\
          S_4 & \mbox { disjoint } \end{array}\right\}
    \end{array} \right]
\]
\end{small}

By Schur complement, so long as $S_0 > 0$ -- we expect to cover at least $t$ edges, $\Z$ is positive semidefinite if and only if
$\Z'$ is:
\begin{small}
\[
  \Z' = \left[ \begin{array}{c|cc}
  		&J, |J|=1&J, |J|=2\\\hline
      I, |I|=1&
      \left\{\begin{array}{cc}
      		S_1 - S_1^2/S_0 & \mbox { if same } \\
      		S_2 - S_1^2/S_0 & \mbox { if different }\end{array}\right\}
      &  B^T \\
      I, |I|=2&B &
      \left\{\begin{array}{cc} 
          S_2 - S_2^2/S_0 & \mbox { if equal } \\
          S_3 - S_2^2/S_0 & \mbox { if 1 vertex shared } \\
          S_4 - S_2^2/S_0 & \mbox { disjoint } \end{array}\right\}
    \end{array} \right]
\]\end{small}

For $B$ given by:
\begin{small}
\[
  B = 
    \left[\begin{array}{c|cc} & J, |J|=1 \\ \hline
    		\multirow{2}{*}{$I, |I|=2$} &
          S_2 - S_1S_2/S_0 & |I \cap J| = 1  \\
          &S_3 - S_1S_2/S_0 & \mbox{ otherwise } \end{array}\right] \]
\end{small}
It is clear that $\Z' = (S_3 - S_4)L + (S_2 - S_4)I + (S_4 - S_2^2/S_0)J$, where $I$ is the identity matrix, $J_{{n \choose 2}}$ is the all-$1$'s matrix on ${n \choose 2}$ rows and columns, and $L$ is the adjacency matrix for the line graph of $K_n$.

Now, so long as $(S_2 - S_1^2/S_0)J_n + (S_1 - S_2)I$ is positive definite, by Schur complement we have that  $\Z'$ is positive semidefinite so long as $Z''$ is for: \[Z'' = (S_2 - S_4 - \frac{\alpha_2 - \alpha_0}{\alpha})I + (S_3 - S_4 - \frac{\alpha_1 - \alpha_0}{\alpha})L + \left(S_4 - \frac{S_2^2}{S_0} - \frac{\alpha_0}{\alpha} + \frac{\beta(2\beta_1 + (n-2)\beta_0)^2}{\alpha^2 + \alpha\beta}\right)J_{{n\choose 2}},\]
with $\alpha = S_1 - S_2$, $\beta = S_2 - S_1^2/S_0$, $\beta_0 = S_3 - S_1S_2/S_0$, $\beta_1 = S_2 - S_1S_2/S_0$,
$\alpha_0 = 4\beta_1\beta_0 = (n-4)\beta_0^2$, $\alpha_1 = 2\beta_1\beta_0 +\beta_1^2 + (n-3)\beta_0^2$, and $\alpha_2 = 2\beta_1 + (n-2)\beta_0^2$.

Now the line graph of $K_n$ is a strongly regular graph with parameters $\left({n \choose 2}, 2n - 4, n - 2, 4\right)$, so it has eigenvalues
$2n - 4$, $n - 4$, and $-2$\footnote{Strictly speaking when $n = 3$, $2n-4$ and $n-4$ are the only eigenvalues, with the rest being $0$.  This does not affect the analysis however.}.
Moreover, the eigenvalue $2n - 4$ has multiplicity $1$ and corresponds to the line of $\mathbb{R}^n$ parallel to the all-$1$'s vector.
Hence $Z''$ is positive semidefinite when:
\begin{align*}
	(S_2 - S_4 - \frac{\alpha_2 - \alpha_0}{\alpha}) + (2n-4)(S_3 - S_4 - \frac{\alpha_1 - \alpha_0}{\alpha}) + {n \choose 2}\left(S_4 - \frac{S_2^2}{S_0} - \frac{\alpha_0}{\alpha} + \frac{\beta(2\beta_1 + (n-2)\beta_0)^2}{\alpha^2 + \alpha\beta}\right) &\geq 0, \\
	(S_2 - S_4 - \frac{\alpha_2 - \alpha_0}{\alpha}) + (n-4)(S_3 - S_4 - \frac{\alpha_1 - \alpha_0}{\alpha}) &\geq 0, \\
	(S_2 - S_4 - \frac{\alpha_2 - \alpha_0}{\alpha}) - 2(S_3 - S_4 - \frac{\alpha_1 - \alpha_0}{\alpha}) &\geq 0.
\end{align*}

Now when $p \geq 2tn^{-1.5}$ we have that all the above eigenvalues are all greater than $0$ in the limit (as $n$ tends to infinity).
Moreover we have that $S_0 \geq 0$ and that from Section 4.3 that $(S_2 - S_1^2/S_0)J_n + (S_1 - S_2)I$ is positive definite as $p \geq 2tn^{-1.5}$.  Hence the level-$1$ Lasserre system still admits an integrality gap of $\sqrt{n}$.

\subsection{Weaker Bounds on the Lasserre System}
\label{sec: general la}
Finally, in this section, we give a construction for a superconstant integrality gap valid for $\Theta(n)$ levels of the Lasserre system.

\begin{lemma}
For $1 \leq l \leq \frac{n}{2}$, there exists a solution to the level-$l$ \iLa\ SDP which admits an integrality gap of $\Theta(n^\frac{1}{2l + 2})$.
\end{lemma} 

Let $k = 2l$.  By Lemma \ref{lem:edge-elim}, the level-$l$ Lasserre slack matrix is positive semidefinite if and only if the level $k$ vertex-subset-only slack matrix is positive semidefinite.
Henceforth we will assume that vectors and matrices originally indexed over $E[K_n] \cup V[K_n]$ are indexed over subsets of $V[K_n]$, by possibly doubling the Lasserre level.

Furthermore, as we are working in the level-$k$ Lasserre system, we will assume that the subsets we are indexing over have at most $k$ elements.

Let $p$, $n$, $y_A$, $\mathcal{Z}$, and $C_{n,m}$ be defined as above (with the above restrictions taken into account), 
and let $M_A = y_A y_A^T$ for $A \subseteq V[K_n], |A| \leq k$, and let $M_m = \sum_{|A| = m} M_A$. Note that $\mathcal{Z}_k$, the level-$k$ vertex-subset-only slack matrix has the following form:
\[ \mathcal{Z}_k = \sum _{m = 0} ^n p^m (1-p)^{n-m} C_{n,m} M_m. \]

Note that
\[(M_m)_{B,D} = {n-|B \cup D|\choose m-|B \cup D|}.\]

Let $Q_m$ for $0 \leq m \leq n$ be a matrix indexed over subsets of $V[K_n]$ of size at most $k$ be defined as:
\[
	(Q_m)_{B,D} = \frac{
    {n-|B|\choose m-|B|}{n-|D|\choose m-|D|}
    }{
    {n \choose m}
    }
 \]

By Schur Complement,
\[\arraycolsep=1.4pt\def\arraystretch{1.5}
	M_m = \left[
    			\begin{array}{ll} 
                {n \choose m} & a \\ 
                     a^T      & B 
                 \end{array}
                \right]  \succeq 
         \left[
    			\begin{array}{ll} 
                {n \choose m} & a \\ 
                     a^T      &  \frac{1}{{n \choose m}} a^Ta
                 \end{array}
                \right] = Q_m.           
 \]
 
Hence, observing that $C_{n, 0} = -t$, we have that
\[
  \mathcal{Z}_k \succeq -t (1-p)^n M_0  + \sum_{m=1}^n p^{m}(1-p)^{n-m} C_{n,m} Q_m,
\]
where $M_0$ is the matrix with a $1$ in the top-left corner.

Note that although $Q_m$ is indexed by subsets of of $V[K_n]$, the value of $(Q_m)_{B,D}$ depends only on $|B|$ and $|D|$.
Hence $Q_m$ only contains $k+1$ unique rows and columns, and by applying symmetric
elementary row and column operations we may eliminate the rest to obtain matrices $P_m$ indexed by $1, 2, 3, \hdots, k + 1$
for $1 \leq m \leq n$, where:
\[(P_m)_{i,j} = (Q_m)_{B, D} \text { for any } |B| = i +1, |D| = j + 1.\]

Hence, for an appropriate choice of $p$, we wish to show that:
\[
  -t (1-p)^n M'_0  + \sum_{m=1}^n p^{m}(1-p)^{n-m} C_{n,m} P_m \succeq 0.
\]

Note that the above follows when:
\[
  \mathcal{Z}'_k = -t (1-p)^n M'_0  + \sum_{m=1}^{k+1} p^{m}(1-p)^{n-m} C_{n,m} P_m \succeq 0.
\]

We will prove, for an appropriate choice of $p$ that $\mathcal{Z}'_k$ is positive definite.  Note that for every $1 \leq q \leq k$,
the $q \times q$ leading principal minor of $\mathcal{Z}'_k$ is positive definite when $\mathcal{Z}'_q$ is positive definite.
Hence, by Sylvester's Criterion, it is sufficient to show that $\det(\mathcal{Z}'_k)$ is positive for arbitrary $k$.

Let $v_m$ be a vector defined as follows, for $1 \leq i \leq k + 1$:
\[(v_m)_i = \frac{{n-i+1 \choose m-i+1}}{{n \choose m}}.\]

Observe that $P_m = {n \choose m} v_m v_m^T$.  Hence, we may rewrite $\mathcal{Z}'_{k}$ as follows:
\[\mathcal{Z'}_{k} = -t (1-p)^n M'_0 + AA^T,\]
where $A^T$ is the matrix whose $m$th row is given by:
\[
(A^T)_m = \left[\sqrt{p^{m} (1-p)^{n-m} C_{n,m} {n \choose m} }\right] v_m^T
\]

Let $B^T$ be the submatrix of $A^T$ we obtain by removing the first column of $A^T$.  By expanding the determinant
of $\mathcal{Z'}_{k}$ along the first column, we obtain that:
\[\det(\mathcal{Z'}_{k}) = \det(AA^T) - t(1-p)^n \det(BB^T).\]

Let $V^T$ be a matrix where:
\[(V^T)_m = v_m^T.\]

We have that:
\[\det(AA^T) = \det(A^T)^2 = \left({\prod_{m=1}^{k+1} {p^{m} (1-p)^{n-m} C_{n,m} {n \choose m}}}\right)\det(V^T)^2.\]

Note that:
\[
    V^T_{ij} = \frac{{n-j+1 \choose i-j+1}}{{n \choose i}} = \frac{{i \choose j-1}}{{n \choose j-1}}.
 \]
Let $W^T$ be such that 
\[
  W^T_{ij} = {i \choose j-1}.
\]

\begin{claim} $\det(W^T) = 1$. \end{claim}
\begin{proof}
Let $W_1, W_2, \hdots W_{k+1}$ be the columns of $W^T$.  Let $D_1$ be the all-zero column vector and for $2 \leq i \leq k + 1$ let $D_i = W_i - D_{i-1}$.  Let $D$ be the matrix whose columns are given by $D_i$ for $1 \leq i \leq k + 1$ and let $W'^T = W^T - D$.  Note that $W'^T_{i,j} = {i \choose j}$.  

Now, $W'^T$ is obtained from $W^T$ by adding and subtracting columns of $W^T$ from each other and $\det(W'^T) = 1$.  Hence $\det(W^T) = 1$, as desired.
\end{proof}

By linearity
\[
  \det(V^T)=
  {\left({\prod_{m=1}^{k+1} {{n \choose m-1}}}\right)}^{-1} \det(W^T) =
  {\left({\prod_{m=1}^{k+1} {{n \choose m-1}}}\right)}^{-1}.
\]

Hence, for an appropriate choice of $p$, in particular, for $p \in o(\frac{1}{n})$, we have that:
\[
  \det(AA^T) = \Theta(p^{\frac{(k+1) (k+2)}{2}} n^{-\frac{(k+1) (k-4)}{2}}).
\]

Now, by Cauchy-Binet:
\[\det(BB^T) = \sum_{m=1}^{k+1} {\det(B_m^T)^2},\]
where $B_m^T$ is the matrix constructed from $B^T$ by removing the $m$th row of $B^T$.

Let $V^T_m$ be a matrix constructed from $V^T$ by removing the first column and $m$th row of $V^T$.
Then, by expanding the determinant,
\[
  \det(B_m^T)^2= O(p^{\frac{(k+1) (k+2)}{2}-m} n^{-m-1+\frac{(k+1) (k+4)}{2}})\det(V_m^T)^2.
\]

Now, 
\[
  \det(V_m^T) = O(n^{-\frac{k (k+1)}{2}}),
\]
and hence, for an appropriate choice of $p$, in particular, for $p \in o(\frac{1}{n})$,
\[
  \det(B_m^T)^2= O(p^{\frac{(k+1) (k+2)}{2}-k-1} n^{-k-2+\frac{(k+1) (k+4)}{2}-k (k+1)})=O(p^{\frac{(k+1) (k+2)}{2}-k-1} n^{-k-2-\frac{(k+1) (k-4)}{2}}).
\]

Now, for $p$ such that $p \in o(n^{-\frac{k+1}{k+2}})$, we have that $\det(BB^T) \in o(\det(AA^T)$.  Hence the Lasserre system
admits a non-constant integrality gap for all levels $1 \leq l \leq \frac{n}{2}$.
%%%%%%%%%%%%%%%%%%%%%%%%%%%%%%%%%%%
%%%%%%%%%%%%%%%%%%%%%%%%%%%%%%%%%%%
%%%%%%%%%%%%%%%%%%%%%%%%%%%%%%%%%%%
%%%%%%%%%%%%%%%%%%%%%%%%%%%%%%%%%%%

\section{Discussion / Open Problems}

The algorithmic significance of our results pose a natural (and classic) open problem, related also to questions on extended formulations; 
%in the line of research that investigates the power of convex programming for tractable  combinatorial problems; 
Does \pvc\ admit a polysize (or tractable) LP or SDP relaxation that has integrality gap no more than 2, even when $t=O(n)$? It is notable that this question has been studied in~\cite{BGKR13} for a generalization of \pvc\ but with no implications to our problem. 
Note also that our strongest IG lower bounds are valid only when $t/n=\epsilon$, for small enough $\epsilon>0$, where $n$ is the number of vertices of the input graph. As a result, another interesting open question is, given $t$ and $n$, find the smallest $r=r(n,t)$ for which the level-$r$ LP or SDP derived by some \lap\ system has integrality gap no more than 2. In particular, can it be that $r=\omega(1)$ when $t\geq n$?

%Finally, an even more challenging task would be to understand the performance of the \la\ system on the \pvc\ polytope. 

Finally, our SDP IG lower bounds make explicit that global distributions of 0-1 assignments can be used to witness solutions to \iSA\ LP tightenings of superconstant integrality gaps. We also demonstrate that it is almost straightforward to show that the same solutions are robust against SDP tightenings of many \lap\ systems except the \iLa\ system. Can the same family of global distributions  fool \iLa\ SDPs when it is also enriched with intuitive and stronger conditions? A generic positive or negative answer would give new insights in understanding the power of the various SDP hierarchies. 

\ignore{
In other words, \iSA\ LPs are at least as strong as SDPs derived by all SDP hierarchies (but the \iLa\ system) when it comes to solutions that are sampled from 0-1 assignments. The latter observation becomes handy, once such distributions can be used to prove large integrality gaps, as we show it is the case for the partial vertex cover polytope. It is therefore natural to ask whether \iLa\ constraints are not stronger than \iSA\ constraints for the same solutions.
}

\paragraph{Acknowledgments:} We would like to thank the anonymous conference and journal referees for their valuable comments.

%Is the LS+ stronger than Halperin's algorithm? \\
%Is our analysis for the SA on cliques tight?  \\
%Check whether the SA solution we give is PSD. Values parameterized by probability $p$ of choosing edge. Does there exist value of $p$ that makes SA solution PSD? \\
%Is there natural convex relaxation with IG no more than 2? what about $2x_e \geq x_i+x_j \geq x_e$? 

\bibliographystyle{plain}
\bibliography{PartialVertexCover}

\end{document}